\documentclass[a4paper,11pt]{article}

\usepackage{jheppub} 

\usepackage[T1]{fontenc} 

\usepackage{graphicx} 
\usepackage{amsthm}
\usepackage{tikz-network}
\usepackage{subfloat}
\usepackage{subfig}

\usepackage{bbm}

\newtheorem{thm}{Theorem}[section]
\newtheorem{prop}[thm]{Proposition}
\newtheorem{defn}[thm]{Definition}

\renewcommand{\v}[1]{\ensuremath{\mathbf{#1}}} 
\newcommand{\abs}[1]{\left| #1 \right|} 
\newcommand{\ket}[1]{| #1 \rangle } 
\newcommand{\bra}[1]{\langle  #1 |} 
\newcommand{\braket}[2]{\langle  #1 \vphantom{#2} | #2 \vphantom{#1} \rangle } 
\newcommand{\ketbra}[2]{|  #1 \vphantom{#2} \rangle \langle #2 \vphantom{#1} | } 
\newcommand{\diracprod}[3]{\left\langle  #1 \vphantom{#2#3} \right|#2 \left| #3 \vphantom{#1#2} \right\rangle } 
\let\baraccent=\= 
\renewcommand{\=}[1]{\stackrel{#1}{=}} 

\def\C{\mathcal{C}}

\def\M{\mathcal{M}}

\def\P{\mathcal{P}}
\def\Q{\mathcal{Q}}
\def\I{\mathcal{I}}

\def\T{\mathcal{T}}
\def\H{\mathcal{H}}

\def\Tr{\text{Tr}}

\def\id{\mathbbm{1}}

\title{Holographic Tensor Networks as Tessellations of Geometry}

\author{Qiang Wen$^{a}$,}

\author{Mingshuai Xu$^{a}$,}

\author{Haocheng Zhong$^{a}$}

\affiliation{$^a$Shing-Tung Yau Center and School of Physics, Southeast University, Nanjing 210096, China}

\emailAdd{wenqiang@seu.edu.cn}
\emailAdd{xumingshuai@seu.edu.cn}
\emailAdd{zhonghaocheng@seu.edu.cn}

\abstract{Holographic tensor networks serve as toy models for the Anti-de Sitter/Conformal Field Theory (AdS/CFT) correspondence, capturing many of its essential features in a concrete manner. However, existing holographic tensor network models remain far from a complete theory of quantum gravity. A key obstacle is their discrete structure, which only approximates the semi-classical geometry of gravity in a qualitative sense. In \cite{Lin:2024dho}, it was shown that a network of partial-entanglement-entropy (PEE) threads, which are bulk geodesics with a specific density distribution, generates a perfect tessellation of AdS space. Moreover, such PEE-network tessellations can be constructed for more highly symmetric geometries using the Crofton formula. In this paper, we assign a quantum state to each vertex in the PEE network and develop several holographic tensor network models: (1) the factorized PEE tensor network, which takes the form of a tensor product of EPR pairs; (2) the HaPPY-like PEE tensor network constructed from perfect tensors; and (3) the random PEE tensor network. In all these models, we reproduce the exact Ryu-Takayanagi formula by showing that the minimal number of cuts along a homologous surface in the network exactly equals the area of that surface.}

\begin{document} 
\maketitle
\flushbottom

\section{Introduction}

The Anti-de Sitter/Conformal Field Theory (AdS/CFT) correspondence \cite{Witten:1998qj,Gubser:1998bc,Maldacena:1997re}, also known as the holographic principle, is a profound framework in theoretical physics that establishes an equivalence between a gravitational theory in a higher-dimensional spacetime and a quantum field theory on its lower-dimensional boundary. It provides a non-perturbative definition of quantum gravity in AdS space in terms of a well-defined quantum field theory without gravity. Moreover, the duality has deep connections to quantum information theory. For example, concepts like entanglement wedge reconstruction (or subregion duality) \cite{Czech:2012bh,Hamilton:2006az,Morrison:2014jha,Bousso:2012sj,Bousso:2012mh,Hubeny:2012wa,Wall:2012uf,Headrick:2014cta,Jafferis:2015del,Dong:2016eik,Cotler:2017erl} reveal how quantum information in a bulk region is encoded in the boundary theory, leading to natural error-correcting properties \cite{Almheiri_2015,Dong:2016eik,Kamal:2019skn,Kang:2019dfi,Akers:2021fut,Xu:2024xbz,Crann:2024gkv}.

The tensor network approach \cite{Swingle:2009bg,Swingle:2012wq,Qi:2013caa,Miyaji:2015fia,Czech:2015kbp,Harlow:2016vwg,Vasseur:2018gfy,Dong:2018seb,Hung:2019zsk,Bao:2019fpq} provides a concrete framework to explore the structure of the AdS/CFT correspondence \footnote{see \cite{Pastawski:2015qua,Hayden:2016cfa,Yang:2015uoa,Qi:2013caa,Bao:2019fpq,Cheng:2022ori,Dong:2023kyr} for an incomplete list of recent progress.}. More explicitly, we start with a space filled with tensors representing quantum states on each site in the space, then we contract the ``legs'' (indices) to connect adjacent tensors thus get a tensor network, which plays the role of the background geometry. The uncontracted legs of tensors in the bulk and the boundary are regarded as the bulk and the boundary degrees of freedom respectively, and the tensor network provides a quantum-error-correcting code that encodes the bulk states into the boundary Hilbert space \cite{Dong:2016eik,Hayden:2016cfa}. Many toy models of tensor network to simulate the AdS/CFT correspondence have been proposed (see also \cite{Chandra:2023dgq,Chen:2024unp,Geng:2025efs,Bao:2024ixc,Kaya:2025vof,Akella:2025owv,Akers:2021pvd,Akers:2022zxr,Akers:2024ixq,Akers:2023obn,Akers:2024wab,Akers:2024pgq} for recent interesting developments), based on different structures of the networks and different quantum states associated to the tensors. One important holographic feature for these toy models is that, the entanglement entropy is given by the minimal number of cuts along a homologous surface in the network, which matches the area of the minimal surface of the celebrated Ryu-Takayanagi (RT) formula \cite{Ryu:2006bv,Hubeny:2007xt,Casini_2011,Lewkowycz:2013nqa,Faulkner:2013ana,Engelhardt:2014gca}.

Nevertheless, in most of these toy models their structures of the tensor network are usually discrete and only capture the features of the background geometry in a qualitative way. There is a big gap between number of cuts in a discrete network and the area of a surface in a smooth Riemannian manifold. Note that the study of continuous tensor network (cTN) models is not new in the literature; see \cite{Verstraete:2010ft} for the initial proposal and \cite{Haegeman:2012uk,Jennings:2012bsr,Tilloy:2018gvo,Tilloy:2021yre} for subsequent developments. Relevant applications in high energy physics can be found in \cite{Shachar:2021vbu} and \cite{Vidal:2008zz,Haegeman:2011uy,Nozaki:2012zj}. In these models, the continuum nature is achieved by taking the infinitesimal limit of the lattice spacing in the underlying lattice configurations. Despite significant progress in approximating genuine continuous quantum field theories, the information about the background geometry is not fully encoded in the tensor network states constructed in this way, leaving the continuum simulation of such theories far from complete.

Recently, a continuous network as a perfect tessellation of the background AdS geometry was constructed using the so-called partial-entanglement-entropy (PEE) threads \cite{Lin:2023rxc,Lin:2024dho}, which are bulk geodesics with a particular density distribution determined by the PEE structure of the boundary CFT. In such a network, the number of cuts along a surface exactly reproduces the area of this surface. In this work, we will develop three toy models of tensor network based on the PEE threads (or the PEE tensor network for short) in $d$-dimensional Poincar\'e AdS space. The first one is the factorized PEE tensor network constructed via tensor product of EPR pairs. The second one is a HaPPY-like (Harlow-Pastawski-Preskill-Yoshida, \cite{Pastawski:2015qua}) tensor network constructed from perfect tensors. It functions as a quantum error-correcting code and exactly reproduces the RT formula for connected regions in the vacuum CFT state. In the third model, following \cite{Hayden:2016cfa,Vasseur:2018gfy,Jia:2020etj,Cheng:2022ori,Dong:2023kyr,Kaya:2025vof} we associate a random state to each bulk site to build the random PEE tensor network, where we can exactly reproduce the RT formula for generic boundary regions. 

The paper is organized as follows. In section \ref{sec:PEE network}, we give an introduction to the network of PEE threads. In section \ref{sec:setup}, we establish the geometric setup for our tensor network models. We introduce the factorized PEE tensor network in section \ref{sec:Factorized} and give an example for the HaPPY-like PEE tensor network in section \ref{sec:HaPPY}, and then introduce the random PEE tensor network in section \ref{sec:Random}. We end the paper with discussions in section \ref{sec:Discussions} and give a geometric description for the factorized PEE tensor network in appendix \ref{sec:RTinfactoried}.

\section{The network of PEE threads}\label{sec:PEE network}

The partial entanglement entropy (PEE) \cite{Wen:2018whg,Wen:2019iyq,Wen:2020ech,Han:2019scu,Han:2021ycp,Kudler-Flam:2019oru} is a measure of entanglement $\I(A,B)$ between any two non-overlapping regions $A$ and $B$ featured by the key property of additivity. It can be determined by the following physical requirements \cite{Wen:2019iyq,Vidal:2014aal}:
\begin{itemize}
	\item[1.] \textit{Additivity}: $\mathcal{I}\left(A,B\cup C\right)=\mathcal{I}\left(A,B\right)+\mathcal{I}\left(A,C\right)$;
	\item[2.] \textit{Normalization}: $\mathcal{I}\left(A,\bar{A}\right)=S_A$;
	\item[3.] \textit{Permutation symmetry}: $\mathcal{I}\left(A,B \right)=\mathcal{I}\left(B,A\right)$;
	\item[4.] \textit{Positivity}: $\mathcal{I}\left(A,B\right)>0$;
	\item[5.] \textit{Upper bounded}: $\mathcal{I}\left(A,B \right)\leq \min\left\{S_A,S_B\right\}$;
	\item[6.]\textit{$\mathcal{I}\left(A,B\right)$ should be invariant under local unitary transformations inside $A$ or $B$};
	\item[7.] \textit{Symmetry: For any symmetry transformation $\mathcal{T}$ under which $\mathcal{T}A=A'$ and $\mathcal{T}B=B'$, we have $\mathcal{I}\left(A,B\right)=\mathcal{I}\left(A',B'\right)$}.
\end{itemize}
Here the regions are all non-overlapping, $A\cup\bar{A}$ makes a pure state and $S_{A}$ denotes the entanglement entropy of $A$. The mutual information that satisfies the above requirements was first studied in \cite{Casini:2008wt} and is called the extensive mutual information (EMI). However, the EMI does not apply to known theories except for two-dimensional massless free fermions. Later, it was realized that if one abandons the assumption that the quantity represents a mutual information \cite{Wen:2019iyq} and imposes the normalization requirement only for spherical regions \cite{Lin:2024dho}, a unique solution of the above requirements exists for the vacuum state of holographic CFTs. This solution, i.e the PEE, is in general not a mutual information
\footnote{Although, in three‑ and higher‑dimensional vacuum CFTs, the PEE and the EMI share the same bilinear formula \eqref{twopointpee}, they are conceptually different. We list the main differences below:
		\begin{itemize}
			\item 	The EMI should be mutual information, which is additional to the seven requirements above. In this case the normalization property follows directly from mutual information and is therefore not an independent requirement. In contrast, the PEE is generally not a mutual information; the normalization property must be properly imposed as a requirement to define it.
			\item The EMI applies only to EMI models where mutual information is additive. So far, the two‑dimensional massless free fermion is the only justified EMI model. Moreover, it was later shown that such models do not exist in higher‑dimensional actual CFTs or in any of their limits (see \cite{Agon:2021zvp}). Most importantly, the EMI does not apply to holographic CFTs. On the other hand, the PEE can be defined for holographic CFTs, which is the foundation of this paper.
			\item In two‑dimensional quantum systems on a line, the PEE has a general formula: an additive linear combination of subsystem entanglement entropies. It is a special type of conditional mutual information, rather than a mutual information, and is therefore completely different from the EMI.
		\end{itemize}
		
		We have listed seven requirements: additivity, normalization, and five others. In fact, the bilinear form of both the EMI and the PEE is already determined by additivity together with the other five requirements, up to a constant coefficient. This bilinear structure exists in the vacuum state of a generic CFT. Interpreting it as the PEE is crucial for applying it to holographic CFTs and reproducing the RT formula.
		
		The concepts of PEE and EMI were developed independently. The study of EMI originated in the context of two‑dimensional free massless fermions, where the mutual information was found to be additive. It was natural to assume that similar theories with additive mutual information might exist in higher dimensions, and these were termed EMI models. Because the EMI is straightforward to compute and applicable to general entangling surfaces, it has been studied in a series of works (see, e.g., \cite{Bueno:2015rda,Casini:2015woa,Bueno:2019mex}) to compute entanglement entropy (EE) in EMI models using the normalization property. In particular, the first paper, \cite{Casini:2008wt}, ruled out holographic CFTs as EMI models, because the Ryu–Takayanagi (RT) formula is inconsistent with the additivity of mutual information. In summary, the applicability of EMI is extremely limited.
		
		The concept of PEE, on the other hand, was motivated by a natural slicing of the entanglement wedge in holography [66] and by the additive entanglement structure conjectured to describe the spatial distribution of entanglement—namely, the entanglement contour function proposed in \cite{Vidal:2014aal}. In one spatial dimension, a general proposal for the PEE \cite{Wen:2018whg,Wen:2020ech} exists, which defines the PEE as an additive linear combination of subsystem entanglement entropies; this was later referred to as the ALC proposal for PEE. In \cite{Wen:2019iyq}, it was shown that the ALC proposal satisfies all seven requirements (including normalization) for generic theories in one spatial dimension, and that it is not a mutual information. The ALC proposal for PEE is completely different from the EMI. They coincide with each other only in the special case of two‑dimensional massless free fermions, which is the only known EMI model.
		
		For higher‑dimensional theories, the bilinear form serves as an ideal candidate for the PEE, provided one abandons the assumption that the quantity represents a mutual information \cite{Wen:2019iyq}. Even so, the bilinear structure is fixed by additivity and the other five requirements up to a constant coefficient, if the normalization requirement is imposed for an arbitrary region, then no solution exists in general theories that satisfies all seven requirements. In our work, we are particularly interested in the PEE in holographic CFTs, and we have recently found that the coefficient can be consistently determined by imposing the normalization requirement only for spherical regions \cite{Lin:2024dho}, thereby ensuring the existence of a solution.
		
		In summary, abandoning the assumption that the bilinear structure describes the mutual information, together with the careful treatment of the normalization requirement, makes the bilinear structure a new quantum information measure (the PEE) in holographic CFT, which is conceptually different from the EMI. Furthermore, it was shown in \cite{Lin:2024dho} that the kinematic space and the Crofton formula for AdS space can be re‑derived based on the geometrization \cite{Lin:2023rxc} of this bilinear structure.
	
}. 
See also \cite{Vidal:2014aal,Wen:2018whg,Wen:2018mev,Wen:2019iyq,Kudler-Flam:2019oru,Wen:2024yny} for other prescriptions to construct the PEE which satisfy the above requirements, and see \cite{Basu:2023wmv,Rolph:2021nan,Lin:2022aqf,Lin:2023orb,Wen:2021qgx,Wen:2022jxr,Wen:2024uwr,Camargo:2022mme} for related discussions on PEE. 

Due to the properties of additivity and permutation symmetry, any PEE $\I (A,B)$ can be further decomposed into a class of PEE between two different sites $\v{x}$ and $\v{y}$, or the two-point PEE $\mathcal{I}\left(\v{x},\v{y}\right)$, in the following way \cite{Wen:2019iyq,Lin:2023rxc}:
\begin{equation}\label{eq:two-point PEE}
	\mathcal{I}\left(A,B\right)=\int_{A}d\v{x}\int_{B}d\v{y}\:\mathcal{I}\left(\v{x},\v{y}\right).
\end{equation}
Moreover, in the vacuum CFT$_{d}$ on a plane the two-point PEE can be fully determined by imposing all the above requirements for the spherical regions \cite{Lin:2023rxc,Lin:2024dho}, i.e.
\begin{equation}\label{twopointpee}
	\mathcal{I}\left(\v{x},\v{y}\right)=\frac{c}{6}\frac{2^{\left(d-1\right)}\left(d-1\right)}{\Omega_{d-2}\left|\v{x}-\v{y}\right|^{2\left(d-1\right)}},
\end{equation}
where $\Omega_{d-2}=2\pi^{\frac{d-1}{2}}/\Gamma\left(\frac{d-1}{2}\right)$ is the area of $\left(d-2\right)$-dimensional unit sphere. 

In the context of AdS/CFT, the vacuum CFT$_{d}$ state on a plane is dual to the Poincar\'e AdS$_{d+1}$. The two-point PEE  $\mathcal{I}\left(\v{x},\v{y}\right)$ was also proposed to have a natural holographic dual \cite{Lin:2023rxc,Lin:2024dho}, which is the density of the bulk geodesics connecting the boundary sites $\v{x}$ and $\v{y}$. The bulk geodesics whose density are determined in such a way are denoted as the PEE threads. These PEE threads can be described by vector fields, with the norm of each vector field capturing the density of threads. For instance, for each point on the AdS boundary, there exists a vector field that describes all the threads emanating from that point (see Fig.\ref{fig:pt-vx} for illustration and see \cite{Lin:2023rxc} for the explicit formula of the vector fields). The full PEE network is then formed by the superposition of all such vector fields (see an example of $d=2$ in Fig.\ref{fig:PEE network}). In this context, the number of geodesics are countless but one can compute the ``flux'' associated with these vector fields. For example, in \eqref{eq:two-point PEE} the right hand side captures the flux of threads connecting boundary regions $A$ and $B$. Additionally, one can compute the total crossing of the threads through a given bulk surface $\Sigma$. It is important to note that the PEE threads are unoriented curves; each time a thread intersects $\Sigma$, it contributes positively to the total crossing count. Therefore, when computing this quantity, which differs from the conventional net flux in that contributions from opposite directions add rather than cancel, the threads should be weighted by the number of times they intersect $\Sigma$. 

\begin{figure}
	\centering
	\subfloat[]{
		\begin{tikzpicture}[scale=0.5]
			\clip (-5,-2) rectangle (5,8);
			\draw[very thick] (-12,0) -- (12,0);
			\draw[] (0,-0.5) node {$\v{x}$};
			\draw[dashed,very thick] (0,0) to (0,8);
			\draw[dashed,thick] (0,0) arc (0:180:2);
			\draw[dashed,thick] (0,0) arc (0:180:3);
			\draw[dashed,thick] (0,0) arc (0:180:4);
			\draw[dashed,thick] (0,0) arc (180:0:3);
			\draw[dashed,thick] (0,0) arc (180:0:4);
			\draw[dashed,thick] (0,0) arc (180:0:2);
			\draw[dashed,thick] (0,0) arc (180:0:10);
			\draw[dashed,thick] (0,0) arc (0:180:10);
			\draw[dashed,thick] (0,0) arc (0:180:6);
			\draw[dashed,thick] (0,0) arc (180:0:6);
			\filldraw[black] (0,0) circle (2pt);
		\end{tikzpicture}
		\label{fig:pt-vx}
	}
	\qquad	\qquad	
	\subfloat[]{	
		\includegraphics[scale=0.15]{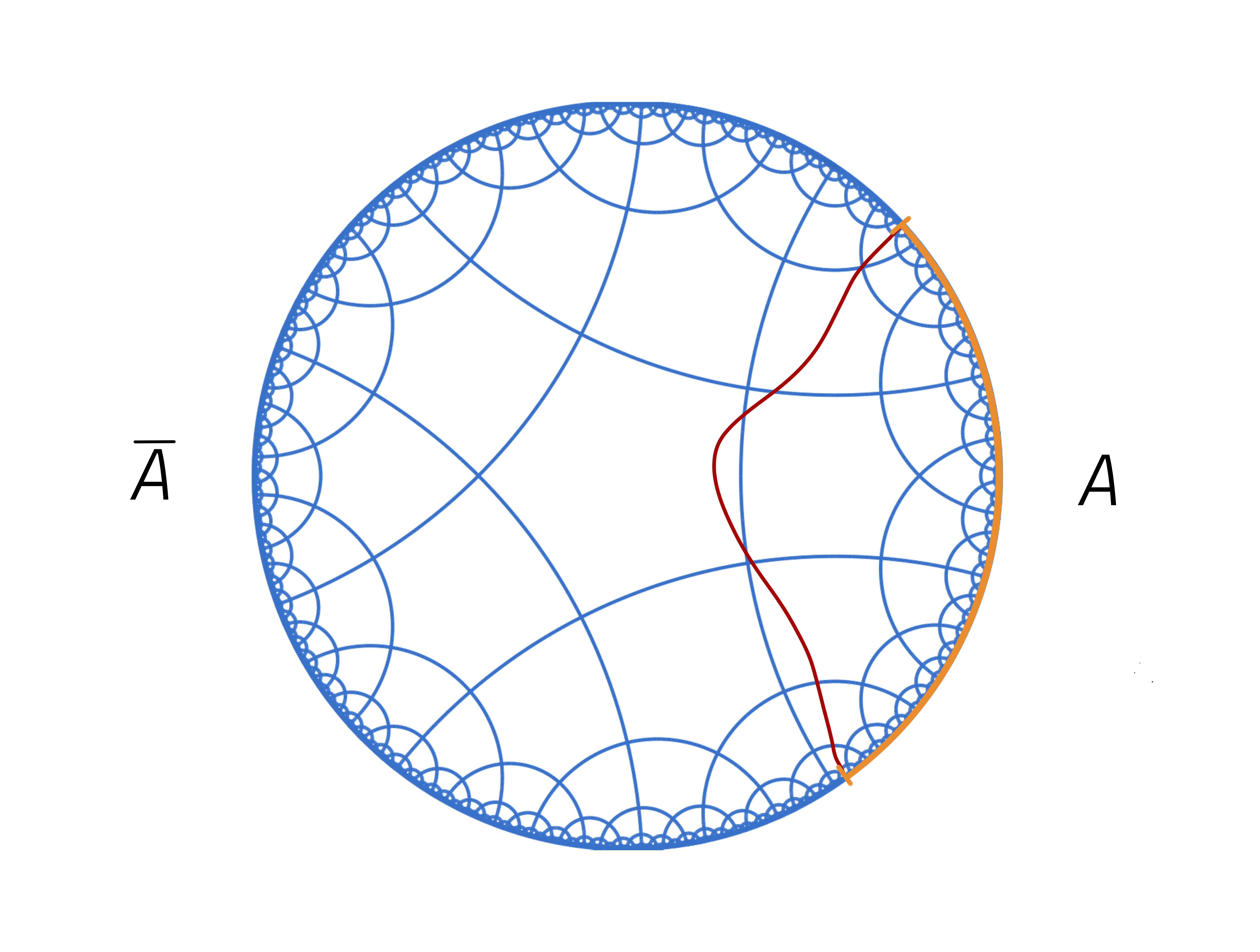}
		\label{fig:PEE network}	}	
	\caption{In the left figure, the black solid line is the AdS boundary, and the dashed lines are geodesics or PEE threads in AdS space. The distribution of the PEE threads emanating from a boundary site $\v{x}$ can be described by a vector field $V_{\v{x}}^{\mu}$, whose integral curves coincide with the PEE threads and whose norm represents the density of these threads. The right figure, taken from \cite{Wen:2024uwr}, shows the PEE network as the superposition of all vector fields $V_{\v{x}}^{\mu}$, resulting in a tessellation of a time slice of AdS$_3$. In the figure, the blue lines represent the PEE threads, and the red line is a surface $\Sigma_A$ homologous to the boundary region $A$.}	
\end{figure}

In discrete networks, concepts such as the number of legs or the number of intersections between a surface and the network are commonly used. We can adopt this language, hence techniques in the discrete models can be used, by defining how much each thread contributes to the flux, or the total crossing each time it passes through $\Sigma$. If we set the Hilbert space dimension on each leg to be $v$, then the contribution from each PEE thread to the flux is simply $\log v$. Here we use $\log v$ instead of $\log_2 v$ to adhere to the convention in high energy physics. Note that the flux or the total crossing, is the essential physical quantity, rather than the number of threads or intersections which depends on our choice of $v$. If we set $\log v=1$ for simplicity, then the flux exactly captures the number of the PEE threads, while the total crossing captures the number of intersections.

Remarkably, it was found that \cite{Lin:2024dho} the PEE network is a perfect tessellation of the static AdS space in the following sense:
\begin{itemize}
	\item \textit{Given an arbitrary $(d-1)$-dimensional surface $\Sigma$ in a time slice of AdS$_{d+1}$, the density of intersections between $\Sigma$ and the bulk PEE network everywhere on $\Sigma$ is given by a constant $\frac{1}{4G}$, where $G=\frac{3}{2c}$ is the Newton's constant}.
\end{itemize}
In the above statement, we used $\log v = 1$ to replace the total crossing with the number of intersections. But this is not essential: any other choice of $\log v$ merely rescales the number or the density of intersections by an overall factor $1/\log v$. Then for any bulk surface $\Sigma$, the area can be reproduced by counting the number $\mathcal{N}\left(\Sigma\right)$ of intersections  between $\Sigma$ and the PEE network,
\begin{equation}\label{PEEnum}
	\frac{\text{Area}(\Sigma)}{4G}=\mathcal{N}\left(\Sigma\right)= \frac{1}{2}\int_{\partial \mathcal{M}}d\v{x}\int_{\partial \mathcal{M}}d\v{y}\:\omega_{\Sigma}\left(\v{x},\v{y}\right)\mathcal{I}\left(\v{x},\v{y}\right).
\end{equation}
Here $\mathcal{M}$ is a constant-time slice in Poincaré AdS$_{d+1}$, and $\partial\mathcal{M}$ is the boundary where the vacuum CFT state resides. The quantity $\omega{_\Sigma}(\mathbf{x},\mathbf{y})$ denotes the number of times the geodesic connecting the boundary sites $\mathbf{x}$ and $\mathbf{y}$ intersects $\Sigma$.

Furthermore, it was shown in \cite{Lin:2024dho} that \eqref{PEEnum} is equivalent to the so-called Crofton formula in integral geometry \cite{Czech:2015qta} applied on a time slice of AdS$_{\text{d+1}}$\footnote{The Crofton formula was first introduced into the study of AdS/CFT by Czech et al. in \cite{Czech:2015qta}, where it was used to reconstruct the length of any bulk curve in a time slice of AdS$_3$ in terms of conditional mutual information. This provides the Crofton formula with a entropic interpretation based on the RT formula.}. To be specific, consider a $d$-dimensional Riemannian manifold $\mathcal{M}$ and an arbitrary $\left(d-1\right)$-dimensional immersed hypersurface $\Sigma$, the Crofton formula states that the area of $\Sigma$ can be evaluated by counting the number of intersections between $\Sigma$ and the geodesics in $\mathcal{M}$, which forms a space $\text{Geod}\left(M\right)$ with the invariant measure $d\Gamma$ which is called the Kinematic measure, i.e.
\begin{equation}\label{Crofton}
	\text{Area}\left(\Sigma\right)=\frac{1}{2}\frac{d-1}{\Omega_{d-2}} \int_{\text{Geod}\left(M\right)}\#\left(\Gamma\cap \Sigma\right)d\Gamma.
\end{equation}
Here, the integrand $\#\left(\Gamma\cap \Sigma\right)$ represents the number of intersections between any geodesic $\Gamma$ in $\text{Geod}\left(M\right)$ and $\Sigma$, and the Kinematic measure $d\Gamma$ specifies the density of geodesics in some parameter space. For the case of a time slice $\mathcal{M}$ in the AdS$_{d+1}$, all the geodesics can be parameterized by its boundary endpoints $\left(\v{x},\v{y}\right)$, and the Kinematic measure $d\Gamma$ is given by \cite{Czech:2016xec,Huang:2015xca,Lin:2024dho},
\begin{equation}\label{Kinematicmeasure}
	\begin{aligned}
		d\Gamma&=\text{det}\left(\frac{\partial^2 \ell\left(\v{x},\v{y}\right)}{\partial \v{x}\partial \v{y}}\right)d\v{x}\wedge d \v{y}=\frac{2^{d-1}}{\left|\v{x}-\v{y}\right|^{2d-2}}d\v{x}\wedge d \v{y},
	\end{aligned}
\end{equation}
where $\ell\left(\v{x},\v{y}\right)$ is the geodesic length of the geodesic connecting $\v{x}$ and $\v{y}$. Notably, the above Kinematic measure for AdS space is proportional to $\mathcal{I}\left(\v{x},\v{y}\right)$ given in \eqref{twopointpee}, hence we can check that \eqref{Crofton} exactly aligns with \eqref{PEEnum}. Note that for a generic manifold, the isometry group may be small, so there is no canonical “invariant” measure on the space of geodesics that is independent of location. The elegant formula \eqref{Crofton} applies only to highly symmetric manifolds such as AdS space, Euclidean space, or the sphere; its generalization to more generic manifolds is highly nontrivial.

One profound implication of the above prescription is that, minimizing the area of a surface $\Sigma_A$ homologous to any boundary region $A$ is equivalent to minimizing the number of intersections between $\Sigma_A$ (see the red line in Fig.\ref{fig:PEE network}) and the PEE network. This further implies that we can exactly reproduce the Ryu-Takayanagi (RT) formula for holographic entanglement entropy by developing proper tensor network models where the entanglement entropy of a boundary region is captured by the minimal number of intersections between a homologous surface and the PEE network \cite{Lin:2023rxc,Lin:2024dho,Wen:2024uwr}. Briefly speaking, the RT formula for the boundary region $A$ can be reformulated as follows,
\begin{equation}\label{eq:entropy from PEE thread flow}
	S_A =\min _{\Sigma_A}\mathcal{N}(\Sigma_A)
	=\min _{\Sigma_A} \frac{\operatorname{Area}\left(\Sigma_A\right)}{4 G}
	=\frac{\operatorname{Area}\left(\gamma_A\right)}{4 G},
\end{equation}
where $\gamma_{A}$ is the famous RT surface.

\section{Geometric setup for bulk tensors}\label{sec:setup}

In a model of tensor network, a vertex at site $\v{z}$ with a set of legs corresponds to a quantum state represented by a local tensor $\T_{u_1u_2\cdots} (\v{z})$, where the index $u_i$ characterizes the leg that is associated to a Hilbert space $\mathcal{H}_{i}$ spanned by the basis vectors $\{\ket{u_i}_{\v{z}}\}$. Then the quantum state in the Hilbert space $\bigotimes_i \mathcal{H}_{i}$ settled at $\v{z}$ is defined by (assuming Einstein summation hereafter),
\begin{equation}\label{vertexstate1}
	\ket{\T(\v{z})} \equiv \T_{u_1u_2\cdots} (\v{z}) \ket{u_1}_{\v{z}}\ket{u_2}_{\v{z}}\cdots.
\end{equation}
For simplicity, we set each of $\{\mathcal{H}_{i}\}$ to be the Hilbert space of a $v$-dimensional qudit, hence $u_i={1,2,\cdots, v}$.

\begin{figure}[h]
	\centering
	{\begin{tikzpicture}[x=0.75pt,y=0.75pt,yscale=-1.2,xscale=1.2]
			
			\path  [fill={rgb, 255:red, 74; green, 144; blue, 226 }  ,fill opacity=0.6] (410.5,231.17) .. controls (417.5,198.17) and (487.5,190.17) .. (501.5,219.17) .. controls (515.5,248.17) and (521,278.67) .. (496.5,298.17) .. controls (495.22,299.19) and (493.86,300.2) .. (492.43,301.2) .. controls (466.31,319.37) and (415.46,333.61) .. (396.5,305.17) .. controls (376.5,275.17) and (403.5,264.17) .. (410.5,231.17) -- cycle ;
			\draw[blue!60,thick]   (471.2,312.8) .. controls (462.6,254.87) and (434.2,216.47) .. (410.5,231.17) ;
			\filldraw[black]  (471.2,312.8) circle (0.5pt);
			\filldraw[black]  (410.5,231.17) circle (0.5pt);
			\draw[blue!60,thick]    (396.5,305.17) .. controls (436.2,221.27) and (503.4,246.07) .. (511,277.67) ;
			\filldraw[black]  (511,277.67) circle (0.5pt);
			\draw[blue!60,thick]    (424.33,320.33) .. controls (412.33,273) and (455,244.33) .. (496.33,213) ;
			\filldraw[black]  (424.5,321) circle (0.5pt);
			\filldraw[black]  (497,212.5) circle (0.5pt);
			\draw  [fill={rgb, 255:red, 255; green, 255; blue, 255 }  ,fill opacity=1 ] (441.2,250.77) .. controls (441.2,245.52) and (445.45,241.27) .. (450.7,241.27) .. controls (455.95,241.27) and (460.2,245.52) .. (460.2,250.77) .. controls (460.2,256.02) and (455.95,260.27) .. (450.7,260.27) .. controls (445.45,260.27) and (441.2,256.02) .. (441.2,250.77) -- cycle ;
			\draw[fill=white] (471.2,312.8) circle[radius=3pt];
			\draw[fill=white] (410.5,231.17) circle[radius=3pt];
			\draw[fill=white] (396.5,305.17) circle[radius=3pt];
			\draw[fill=white] (511,277.67) circle[radius=3pt];
			\draw[fill=white] (424.5,321) circle[radius=3pt];
			\draw[fill=white] (497,212.5) circle[radius=6pt];
			
			\draw[fill=black] (470,233) circle[radius=1.5pt];
			\draw[fill=black] (473,230.6) circle[radius=1.5pt];
			
			\draw[] (451,251) node {$\v{z}$};
			\draw[] (497,212.5) node {$\v{z'}$};
			\draw[] (470,220.5) node {$\ket{\v{z}\v{z'}}$};
			
		\end{tikzpicture}
		
		\caption{A tensor state located at the bulk site $\v{z}$ (hollow white dot in the center). The legs connecting $\v{z}$ is actually continuously distributed, which can be described by a vector field as in \cite{Lin:2023rxc}. Only six legs are exhibited, while other legs are collectively represented by the blue region. The contraction between tensors at $\v{z}$ and at adjacent site $\v{z'}$ is realized by projecting the total state into a maximally entangled state $\ket{\v{z}\v{z'}}$ along the PEE thread in between.}
		\label{fig:bulk tensor}}
\end{figure}

Now we develop the PEE tensor networks. Consider any site $\v{z}$ in a time slice $\M$ of the bulk space, for any PEE thread emanating from it to the boundary, we assign one leg that lies along this PEE thread, hence build a one-to-one correspondence between the PEE threads and the legs connecting to $\v{z}$ (see Fig.\ref{fig:bulk tensor}). In the case of Poincar\'e AdS, any PEE thread emanating from $\v{z}$ intersects the boundary with a unique boundary site $\v{x}$. We denote this thread as $C_{\v{z}}(\v{x})$, and $\v{x}$ can be used to parameterize all the threads or legs associated to $\v{z}$. Let us denote the Hilbert space and basis vectors on each leg as $\mathcal{H}_\v{x}$ and $\ket{k_{\v{x}}}_{\v{z}}$, such that the state \eqref{vertexstate1} is rewritten as 
\begin{equation}\label{eq:bulk local tensor 2}
		\ket{\T(\v{z})}\equiv \T_{\v{k}}(\v{z})\left(\bigotimes_{\v{x}\in bdy}\ket{k_{\v{x}}}_{\v{z}}\right),
\end{equation}
where $\v{k}$ denotes a collection of indices $\{k_{\v{x}}\}$, and $\{\ket{k_{\v{x}}}\}$ forms an orthonormal basis for the Hilbert space associated to the index $k_{\v{x}}$.

Given any adjacent site $\v{z'}$ of $\v{z}$, there is a unique PEE thread connecting them, which picks up a pair of legs associated to $\v{z}$ and $\v{z'}$ that lie within this thread (see Fig.\ref{fig:bulk tensor}). Then we contract all such pairs of legs to glue adjacent sites along the PEE threads, and finally get a quantum state $\ket{\Psi}$ for the boundary open (uncontracted) legs,
\begin{equation}\label{contraction0}
	\ket{\Psi}=	\left(\bigotimes_{\v{z},\v{z'} \in \mathcal{M}} \bra{\v{z}\v{z'}}\right)\left(\bigotimes_{\v{z}\in \mathcal{M}} \ket{\T(\v{z})}\right)
\end{equation}
where $ \ket{\v{z}\v{z'}}=\frac{1}{\sqrt{v}}\sum_{k=1}^{k=v}\ket{k}_{\v{z}}\ket{k}_{\v{z'}}$ is a maximally entangled state defined on any pair of legs connecting the nearby sites $\v{z}$ and $\v{z'}$. Such construction is also referred to as
	projected entangled pair states (PEPS)\cite{Verstraete:2003wvw,Verstraete:2004cf,Verstraete:2008cex,Cirac:2020obd} in condensed matter physics. This projection can be understood as the gluing of nearby geometries through the imposition of identical boundary conditions when performing the path integral in gravitational theories. These contractions connect all the legs into complete PEE threads (see Fig.\ref{fig:tensor at a PEE thread}) and eventually we get the PEE tensor networks.

\begin{figure}[h]
	\centering
	\begin{tikzpicture}[x=0.75pt,y=0.75pt,yscale=-2.3,xscale=2.3]
		
		\draw [blue!50,thick]
		(300,150) .. controls (300,150) and (300,150) .. (300,150)
		.. controls (300,122.39) and (277.61,100) .. (250,100)
		.. controls (222.39,100) and (200,122.39) .. (200,150);
		\draw[thick]
		(320,148.5)--(180,148.5);
		\draw[thick,red]
		(200,150) arc (180:190:50);
		\draw[thick,red]
		(300,150) arc (360:350:50);
		\draw  [draw opacity=0] (202.08,115.06) .. controls (206.92,118.71) and (210.61,123.8) .. (212.53,129.69) -- (184,139) -- cycle ; \draw [blue!50,thick] (202.08,115.06) .. controls (206.92,118.71) and (210.61,123.8) .. (212.53,129.69) ;  
		\draw  [draw opacity=0] (199.72,120.02) .. controls (207.15,120.45) and (213.88,123.53) .. (218.97,128.33) -- (197.91,150.68) -- cycle ; \draw  [blue!50,thick] (199.72,120.02) .. controls (207.15,120.45) and (213.88,123.53) .. (218.97,128.33) ;  
		\draw  [draw opacity=0] (225.85,113.12) .. controls (229.66,106.06) and (234.56,99.67) .. (240.32,94.17) -- (292.7,149.1) -- cycle ; \draw  [blue!50,thick] (225.85,113.12) .. controls (229.66,106.06) and (234.56,99.67) .. (240.32,94.17) ;  
		\draw  [draw opacity=0] (222.02,100.57) .. controls (229.22,101.68) and (235.79,104.68) .. (241.21,109.04) -- (215.82,140.61) -- cycle ; \draw  [blue!50,thick] (222.02,100.57) .. controls (229.22,101.68) and (235.79,104.68) .. (241.21,109.04) ;  
		\draw  [draw opacity=0] (259.06,108.61) .. controls (263.9,104.07) and (269.82,100.68) .. (276.38,98.86) -- (287.6,139.09) -- cycle ; \draw [blue!50,thick]  (259.06,108.61) .. controls (263.9,104.07) and (269.82,100.68) .. (276.38,98.86) ;  
		\draw  [draw opacity=0] (258.51,95.04) .. controls (265.03,99.05) and (269.89,105.49) .. (271.85,113.08) -- (242.8,120.6) -- cycle ; \draw  [blue!50,thick] (258.51,95.04) .. controls (265.03,99.05) and (269.89,105.49) .. (271.85,113.08) ;  
		\draw  [draw opacity=0] (283.5,128.29) .. controls (288.37,123.18) and (295.03,119.79) .. (302.47,119.12) -- (305.2,149) -- cycle ; \draw [blue!50,thick]  (283.5,128.29) .. controls (288.37,123.18) and (295.03,119.79) .. (302.47,119.12) ;  
		\draw  [draw opacity=0] (288.03,131.35) .. controls (289.77,125.21) and (292.55,119.52) .. (296.17,114.46) -- (341.03,146.45) -- cycle ; \draw [blue!50,thick] (288.03,131.35) .. controls (289.77,125.21) and (292.55,119.52) .. (296.17,114.46) ;  
		\draw [draw opacity=0] (194.58,146.02) .. controls (195.61,144.91) and (197.08,144.22) .. (198.72,144.22) .. controls (201.85,144.22) and (204.38,146.76) .. (204.38,149.89) .. controls (204.38,149.93) and (204.38,149.98) .. (204.38,150.02) -- (198.72,149.89) -- cycle ; \draw [blue!50,thick]  (194.58,146.02) .. controls (195.61,144.91) and (197.08,144.22) .. (198.72,144.22);
		\draw  [draw opacity=0] (202.78,145.94) .. controls (203.77,146.96) and (204.38,148.36) .. (204.38,149.89) .. controls (204.38,149.93) and (204.38,149.98) .. (204.38,150.02) -- (198.72,149.89) -- cycle ; \draw  [red,draw opacity=1 ] (202.78,145.94) .. controls (203.77,146.96) and (204.38,148.36) .. (204.38,149.89) .. controls (204.38,149.93) and (204.38,149.98) .. (204.38,150.02) ; 
		\draw  [draw opacity=0] (196.82,150.23) .. controls (197.11,145.6) and (200.96,141.95) .. (205.65,141.97) .. controls (205.73,141.97) and (205.82,141.97) .. (205.9,141.98) -- (205.61,150.79) -- cycle ; \draw [blue!50,thick]  (196.82,150.23) .. controls (197.11,145.6) and (200.96,141.95) .. (205.65,141.97) .. controls (205.73,141.97) and (205.82,141.97) .. (205.9,141.98) ;  
		\draw  [red,draw opacity=1 ] (196.82,150.23) .. controls (196.9,148.88) and (197.28,147.63) .. (197.9,146.52) ;  
		\filldraw[black] (196.82,150) circle (0.6pt); 
		\draw  [draw opacity=0] (296.89,150) .. controls (296.89,149.95) and (296.89,149.9) .. (296.89,149.85) .. controls (296.88,146.27) and (299.78,143.35) .. (303.37,143.33) .. controls (304.33,143.33) and (305.24,143.53) .. (306.06,143.9) -- (303.39,149.83) -- cycle ; \draw [blue!50,thick]  (296.89,150) .. controls (296.89,149.95) and (296.89,149.9) .. (296.89,149.85) .. controls (296.88,146.27) and (299.78,143.35) .. (303.37,143.33) .. controls (304.33,143.33) and (305.24,143.53) .. (306.06,143.9) ;  
		\draw  [red,draw opacity=1 ] (296.89,150) .. controls (296.89,149.95) and (296.89,149.9) .. (296.89,149.85) .. controls (296.88,148.65) and (297.21,147.51) .. (297.78,146.54) ; 
		\filldraw[black] (296.89,150) circle (0.6pt);
		\draw  [draw opacity=0] (293.77,143) .. controls (294.13,142.94) and (294.51,142.91) .. (294.9,142.91) .. controls (299,142.9) and (302.34,146.22) .. (302.36,150.33) .. controls (302.36,150.34) and (302.36,150.34) .. (302.36,150.35) -- (294.91,150.35) -- cycle ; \draw [blue!50,thick]  (293.77,143) .. controls (294.13,142.94) and (294.51,142.91) .. (294.9,142.91) .. controls (299,142.9) and (302.34,146.22) .. (302.36,150.33) .. controls (302.36,150.34) and (302.36,150.34) .. (302.36,150.35) ;  
		\draw  [red,draw opacity=1 ] (301.46,146.8) .. controls (302.03,147.85) and (302.35,149.05) .. (302.36,150.33) .. controls (302.36,150.34) and (302.36,150.34) .. (302.36,150.35) ;  
		\draw  [fill={rgb, 255:red, 255; green, 255; blue, 255 }  ,fill opacity=1 ] (205.56,122.03) .. controls (205.56,120.23) and (207.01,118.78) .. (208.81,118.78) .. controls (210.6,118.78) and (212.06,120.23) .. (212.06,122.03) .. controls (212.06,123.82) and (210.6,125.28) .. (208.81,125.28) .. controls (207.01,125.28) and (205.56,123.82) .. (205.56,122.03) -- cycle ;
		\draw   [fill={rgb, 255:red, 255; green, 255; blue, 255 }  ,fill opacity=1 ](263.33,103.14) .. controls (263.33,101.34) and (264.79,99.89) .. (266.58,99.89) .. controls (268.38,99.89) and (269.83,101.34) .. (269.83,103.14) .. controls (269.83,104.93) and (268.38,106.39) .. (266.58,106.39) .. controls (264.79,106.39) and (263.33,104.93) .. (263.33,103.14) -- cycle ;
		\draw  [fill={rgb, 255:red, 255; green, 255; blue, 255 }  ,fill opacity=1 ] (228.44,103.81) .. controls (228.44,102.01) and (229.9,100.56) .. (231.69,100.56) .. controls (233.49,100.56) and (234.94,102.01) .. (234.94,103.81) .. controls (234.94,105.6) and (233.49,107.06) .. (231.69,107.06) .. controls (229.9,107.06) and (228.44,105.6) .. (228.44,103.81) -- cycle ;
		\draw  [fill={rgb, 255:red, 255; green, 255; blue, 255 }  ,fill opacity=1 ] (288.44,122.47) .. controls (288.44,120.68) and (289.9,119.22) .. (291.69,119.22) .. controls (293.49,119.22) and (294.94,120.68) .. (294.94,122.47) .. controls (294.94,124.27) and (293.49,125.72) .. (291.69,125.72) .. controls (289.9,125.72) and (288.44,124.27) .. (288.44,122.47) -- cycle ;
		\draw [fill={rgb, 255:red, 255; green, 255; blue, 255 }  ,fill opacity=1 ]  (196.89,144.25) .. controls (196.89,142.46) and (198.34,141) .. (200.14,141) .. controls (201.93,141) and (203.39,142.46) .. (203.39,144.25) .. controls (203.39,146.04) and (201.93,147.5) .. (200.14,147.5) .. controls (198.34,147.5) and (196.89,146.04) .. (196.89,144.25) -- cycle ;
		\draw [fill={rgb, 255:red, 255; green, 255; blue, 255 }  ,fill opacity=1 ]  (296.44,144.03) .. controls (296.44,142.23) and (297.9,140.78) .. (299.69,140.78) .. controls (301.49,140.78) and (302.94,142.23) .. (302.94,144.03) .. controls (302.94,145.82) and (301.49,147.28) .. (299.69,147.28) .. controls (297.9,147.28) and (296.44,145.82) .. (296.44,144.03) -- cycle ;
		\filldraw[black] (300,150) circle (0.6pt);
		\filldraw[black] (200,150) circle (0.6pt);
		\filldraw[black] (302.36,150) circle (0.6pt);
		\filldraw[black] (204.38,150) circle (0.6pt);
		\draw[] (200.4,144) node {$\v{x}$};
		\draw[] (300,144.5) node {$\v{y}$};
		\draw[] (200.4,155) node {$\uparrow$};
		\draw[] (300,155) node {$\uparrow$};
		\draw[] (200.4,162) node {$\ket{a}_{\v{x}}$};
		\draw[] (300,162) node {$\ket{b}_{\v{y}}$};
		\draw[] (300,110) node {$\sim T_{ab}(\v{x},\v{y})$};
		\draw[] (209,122.5) node {$\v{z}_1$};
		\draw[] (292,122.5) node {$\v{z}_n$};
		\draw[] (232,104.5) node {$\v{z}_2$};
		\draw[] (269,104) node {$\v{z}_{n-1}$};
		\draw[] (250,104) node {$\cdots$};
	\end{tikzpicture}	
	\caption{ Contraction of bulk tensors along a PEE thread results in a PEE thread with two open legs of the boundary sites tangent to the thread. In the factorized PEE tensor network, the thread represents a two-qudit state $T_{ab}(\v{x},\v{y})\ket{a}_{\v{x}}\ket{b}_{\v{y}}$. }
	\label{fig:tensor at a PEE thread}
\end{figure}

\section{Factorized PEE tensor network}\label{sec:Factorized}

Let us consider a pair of PEE threads which makes a complete geodesic passing through $\v{z}$, and we denote it as $C_{\v{z}}(\v{x},\v{x'})\equiv C_{\v{z}}(\v{x})\cup C_{\v{z}}(\v{x'})$. Correspondingly we can divide the legs of $\v{z}$ into pairs and rewrite the quantum state \eqref{eq:bulk local tensor 2} as, 
\begin{equation}\label{eq:bulk local tensor 3}
		\ket{\T(\v{z})}=\T_{\v{a}\v{b}}(\v{z})\left(\bigotimes_{\C_{\v{z}}(\v{x},\v{x'})}\ket{a_{\v{x}}}_{\v{z}}\ket{b_{\v{x'}}}_{\v{z}}\right),
	\end{equation}
	where $\v{a}$ ($\v{b}$) denotes a collection of indices $\{a_{\v{x}}\}$ ($\{b_{\v{x'}}\}$). We define the factorized PEE tensor network in the following:
\begin{itemize}
	\item Firstly, we set the tensor $\T_{\v{a}\v{b}}(\v{z})$ to be a tensor product of two-index tensors, each of which is associated to a pair of legs lying within the same complete PEE thread, such that the quantum state \eqref{eq:bulk local tensor 2} becomes a product of states associated to each pair of legs:
	\begin{equation}\label{eq:bulk local tensor}
			\ket{\T(\v{z})}=\bigotimes_{\C_{\v{z}}(\v{x},\v{x'})}\T_{ab}(\v{z})\ket{a_{\v{x}}}_{\v{z}}\ket{b_{\v{x'}}}_{\v{z}}\,.
	\end{equation}
	
	\item Furthermore, we require the state $\T_{ab}(\v{z})\ket{a_{\v{x}}}_{\v{z}}\ket{b_{\v{x'}}}_{\v{z}}$ to be a pair of maximally entangled qudits, i.e. EPR pairs.
\end{itemize} 

Then we contract all the internal legs to get the tensor network. Note that only the legs of adjacent sites are contracted along the PEE thread in between, hence the contractions are always along the PEE threads. More explicitly, let us consider a PEE thread $\C(\v{x},\v{y})$ connecting boundary sites $\v{x}$ and $\v{y}$, and write $\C(\v{x},\v{y})=\{\v{x},\v{z}_1,\dots,\v{z}_n, \v{y}\}$ as a set of consecutive adjacent bulk sites with $n\rightarrow\infty$, we contract all the two-index tensors on $\C(\v{x},\v{y})$ to get,
\begin{equation}\label{eq:tensor on PEE thread}
	T_{ab}(\v{x},\v{y})\equiv\T_{ac_1}(\v{x})\T_{c_1c_2}(\v{z}_1)\dots \T_{c_{n}c_{n+1}}(\v{z}_{n})\T_{c_{n+1}b}(\v{y}),
\end{equation}
which is also a two-index tensor representing a pair of maximally entangled qudits. The contraction gives us a complete PEE thread with the two open legs tangent to $\C(\v{x},\v{y})$ in the EPR state $\ket{\psi}_{\v{x}\v{y}}=T_{ab}(\v{x},\v{y})\ket{a}_{\v{x}}\ket{b}_{\v{y}}$, (see Fig.\ref{fig:tensor at a PEE thread}). Note that, since $\v{x}$ and $\v{y}$ are boundary sites, we cannot use them to parameterize their legs since by default there are numerous qudits with open leg lying in the same boundary site. Just keep in mind that, here $\ket{a}_{\v{x}}\ket{b}_{\v{y}}$ only denotes the state of the pair of open legs that is tangent to the corresponding PEE thread $\mathcal{C}(\v{x},\v{y})$ respectively. After the contraction along all the PEE threads, we get a tensor network representing the quantum state for all the boundary open legs:
\begin{align}\label{boundarystate}
	\ket{\Psi}=\bigotimes_{_{\C(\v{x},\v{y})}} T_{ab}(\v{x},\v{y})\ket{a}_{\v{x}}\ket{b}_{\v{y}}\,,
\end{align}
which is a product state of EPR pairs. The reduced density matrix for any one of the qudits in the EPR pairs satisfies:
\begin{align}\label{eq:unitary tensor}
	\rho_{\v{x}}=&T_{ab}(\v{x},\v{y})T^{*}_{cb}(\v{x},\v{y})\left(\ket{a}\bra{c}\right)_{\v{x}}=\frac{\id_{\v{x}}}{v},
\end{align}
which demonstrates the unitarity of $T_{ab}(\v{x},\v{y})$.

Now we consider any boundary region $A$ in the state \eqref{boundarystate}, and compute the entanglement entropy for $A$. It is convenient to classify all the PEE threads into three classes: (a) $\mathcal{C}^{A}$: threads that anchor only on $A$; (b) $\mathcal{C}^{\bar{A}}$: threads that anchor only on $\bar{A}$; (c) $\mathcal{C}^{A\bar{A}}$: threads that anchor on both $A$ and $\bar{A}$ at each end. See the Fig.\ref{fig:various PEE threads} for illustration. Let us rewrite $\ket{\Psi}=\ket{T^{A}}\otimes\ket{T^{\bar{A}}}\otimes\ket{T^{A\bar{A}}}$ where
\begin{equation}
	\begin{aligned}
		\ket{T^{A}}&\equiv \bigotimes_{\mathcal{C}(\v{x},\v{y})\in \mathcal{C}^A} T^{(A)}_{ab}(\v{x},\v{y})\ket{a}_{\v{x}}\ket{b}_{\v{y}},
		\\
		\ket{T^{\bar{A}}}&\equiv\bigotimes_{\mathcal{C}(\v{x},\v{y})\in \mathcal{C}^{\bar{A}}} T^{(\bar{A})}_{ab}(\v{x},\v{y})\ket{a}_{\v{x}}\ket{b}_{\v{y}},
		\\
		\ket{T^{A\bar{A}}}&\equiv\bigotimes_{\mathcal{C}(\v{x},\v{y})\in \mathcal{C}^{A\bar{A}}} T^{(A\bar{A})}_{ab}(\v{x},\v{y})\ket{a}_{\v{x}}\ket{b}_{\v{y}}.
	\end{aligned}
\end{equation}

Then we define pure states $\rho^{(A)}=\ketbra{T^{A}}{T^{A}},~\rho^{(\bar{A})}=\ketbra{T^{\bar{A}}}{T^{\bar{A}}},~\rho^{(A\bar{A})}=\ketbra{T^{A\bar{A}}}{T^{A\bar{A}}}$, and we have $\rho=\ketbra{\Psi}{\Psi}=\rho^{(A)}\otimes \rho^{(\bar{A})}\otimes\rho^{(A\bar{A})}$,
which implies that
\begin{equation}
	\rho_{A}=\Tr_{\bar{A}}\,\rho=\rho^{(A)}_{A}\otimes \rho^{(\bar{A})}_{A}\otimes\rho^{(A\bar{A})}_{A}\,.
\end{equation}
Note that $\rho^{(A)}_{A}=\rho^{(A)}=\ketbra{T^{A}}{T^{A}}$ is pure, and we also have $\rho^{(\bar{A})}_{A}=\Tr_{\bar{A}}\rho^{(\bar{A})}=\Tr_{\bar{A}}\ketbra{T^{\bar{A}}}{T^{\bar{A}}}=\braket{T^{\bar{A}}}{T^{\bar{A}}}=1$, hence $S(\rho^{(A)}_{A})=S(\rho^{(\bar{A})}_{A})=0$. As for $\rho^{(A\bar{A})}_{A}$, according to the unitarity of $T_{ab}(\v{x},\v{y})$ in \eqref{eq:unitary tensor}, we have
\begin{align}
	\rho^{(A\bar{A})}_{A}=&\bigotimes_{\mathcal{C}(\v{x},\v{y})\in \mathcal{C}^{A\bar{A}}}\left(T(\v{x},\v{y})T^\dagger(\v{x},\v{y})\right)
	\cr
	=&\bigotimes_{\mathcal{C}(\v{x},\v{y})\in \mathcal{C}^{A\bar{A}}}\frac{\id_{\v{x}}}{v}.
\end{align}
In order to match the case of AdS/CFT, the total flux of PEE threads connecting $A$ and $\bar{A}$ should match to $\mathcal{I}(A,\bar{A})$. Again, let us set $\log v=1$, then the number of PEE threads in $\mathcal{C}^{A\bar{A}}$ coincides with $\I(A,\bar{A})$, and the entanglement entropy between $A$ and $\bar{A}$ is given by:
\begin{equation}\label{eq:entropy factorized PEE}
	S_{A}=S(\rho^{(A\bar{A})}_{A})=\I(A,\bar{A}).
\end{equation}

It has been demonstrated that \cite{Lin:2024dho}, given a static boundary interval $A$ (or any spherical region in higher dimensions) and the corresponding RT surface $\gamma_{A}$, the PEE threads in $\mathcal{C}^{A\bar{A}}$ are the only threads that intersect with $\gamma_{A}$, and each thread only intersects once. This means the number of intersections between $\gamma_{A}$ and the PEE network is exactly $\mathcal{I}(A,\bar{A})$. Furthermore, since the area of $\gamma_{A}$ is computed by counting the number of intersections \eqref{PEEnum}, we have
\begin{align}
	S_{A}=\I(A,\bar{A})=\frac{\text{Area}(\gamma_{A})}{4G}\,,
\end{align}
which reproduces the RT formula for any static boundary single intervals and spherical regions.  See the appendix \ref{sec:RTinfactoried} for more details for the bulk geometric picture, and another derivation for the RT formula by constructing isometric mappings from $\gamma_{A}$ to the boundary.

Nevertheless, for disconnected and non-spherical connected boundary regions, the homologous surfaces $\Sigma_A$ that only intersect with the threads in $\mathcal{C}^{A\bar{A}}$ no longer exists. Furthermore, the RT surface $\gamma_{A}$ intersects not only with all the threads in $\mathcal{C}^{A\bar{A}}$, but also with a subset of the threads in $\mathcal{C}^{A}$ or $\mathcal{C}^{\bar{A}}$ \cite{Lin:2024dho}, which means $\I(A,\bar{A})<\frac{Area(\gamma_{A})}{4G}$. For these reasons we fail to reproduce the RT formula for generic boundary regions in this model.

\section{HaPPY-like PEE tensor network}\label{sec:HaPPY}

Besides regarding the tensors as tensor product of two-index tensors, we can also construct a HaPPY-like PEE tensor network following \cite{Pastawski:2015qua}, where we associate a perfect tensor on each site in the PEE network. Firstly, let us recall the definition of perfect tensor as follows:
\begin{defn}
	A $2n$-index tensor $T_{a_1a_2\dots a_{2n}}$
	is a perfect tensor if, for any bipartition
	of its legs into a set $D$ and complementary set $D^c$ with $\abs{D}\leq\abs{D^c}$, $T$ is proportional
	to an isometric tensor from $D$ to $D^c$, where $\abs{\cdot}$ denotes the number of legs in this set.
\end{defn}
The definition implies that if we assign a flow configuration to a perfect tensor, the tensor always represents an isometric map as long as the number of incoming legs is less than or equal to the number of outgoing legs.

\begin{figure}[hbtp]
	\subfloat[]{
		\includegraphics[width=0.3\linewidth]{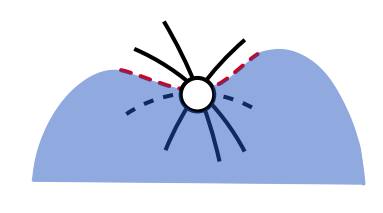}
		\label{fig:case 1}
	}
	\hfill
	\subfloat[]{
		\includegraphics[width=0.32\linewidth]{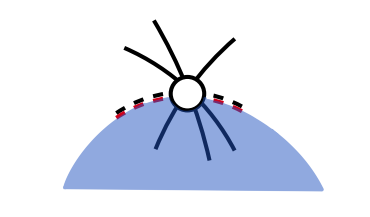}
		\label{fig:case 2}
	}
	\hfill
	\subfloat[]{
		\includegraphics[width=0.23\linewidth]{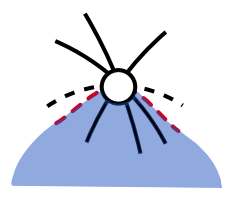}
		\label{fig:case 3}
	}
	\caption{Three local configurations for the legs emanating from one bulk site (the hollow circle) on the homologous surface $\Sigma_A$ (red dashed line) in the PEE network. Here, $\mathcal{W}_{\Sigma_A}$ is marked as the blue region, and the PEE thread tangent to $\Sigma_A$ at this bulk site is marked as the black dashed line.}
	\label{fig:3 cases}
\end{figure}

Let us start from a generic tensor network with no open legs in the bulk, and consider a connected boundary region $A$. Furthermore, let us cut the tensor network open along a homologous surface $\mathcal{C}$ whose boundary matches the boundary of $A$, such that the tensor network is divided into two parts $\P$ and $\Q$, which are maps from $\mathcal{C}$ to $A$ and from $\mathcal{C}$ to $\bar{A}$ respectively. The boundary state can be written as
	\begin{equation}\label{eq:holo state}
		\ket{\Psi}=\P_{\v{a}\v{c}}\Q_{\v{b}\v{c}}\ket{\v{a}}_{A}\ket{\v{b}}_{\bar{A}}=\ket{\P_\v{c}}_{A}\ket{\Q_{\v{c}}}_{\bar{A}},
	\end{equation}
	where $\{\ket{\v{a}}\},~\{\ket{\v{b}}\}$ runs over the complete bases for $A$ and $\bar{A}$, and the index $\v{c}$ collects degrees of freedom associated to $\mathcal{C}$. Tracing $\bar{A}$ we get the reduced density matrix,
	\begin{equation}
		\rho_A=\braket{\Q_\v{c'}}{\Q_{\v{c}}}\ket{\P_\v{c}}\bra{\P_\v{c'}},
	\end{equation}
	whose corresponding entanglement entropy $S_A$ is upper bounded by the log of the Hilbert space dimension of $\mathcal{C}$, i.e. $\mathcal{N}(\mathcal{C}) \log v$ where $\mathcal{N}(\mathcal{C})$ is the number of intersections between $\mathcal{C}$ and the network. Then we get the minimal upper bound by choosing the homologous surface with minimal number of cuts, which we also denote as $\gamma_{A}$, hence
	\begin{equation}\label{entropybound}
		S_A\leq \mathcal{N}(\gamma_A) \log v.
	\end{equation}
	As demonstrated in \cite{Pastawski:2015qua}, the above inequality is saturated if and only if the two tensors $\P$ and $\Q$, obtained by cutting the entire bulk tensor along $\gamma_{A}$, are isometric. We now apply this discussion to our PEE network constructed from perfect tensors, and again set $\log v = 1$ so that $\mathcal{N}(\gamma_{A}) = \text{Area}(\gamma_{A})/{4G}$. To reproduce the RT formula, it suffices to show that the mapping $\P$ is an isometry from $\gamma_A$ to $A$, and $\Q$ is an isometry from $\gamma_{A}$ to $\bar{A}$.

Now we prove the isometric property of $\P$ and $\Q$ using an analogous greedy algorithm proposed in \cite{Pastawski:2015qua}. The procedure is decomposed into the following steps:
\begin{enumerate}
	\item Given any connected boundary region $A$, let us consider a homologous surface $\Sigma_A$ which cuts the PEE tensor network open, and we denote the region enclosed by $A$ and $\Sigma_A$ as $\mathcal{W}_{\Sigma_A}$ (we do not include the sites on $\Sigma_A$ into $\mathcal{W}_{\Sigma_A}$). In Fig.\ref{fig:3 cases}, we mark $\mathcal{W}_{\Sigma_A}$ as the blue region and $\Sigma_A$ as the dashed red line. Note that, we start from a $\mathcal{W}_{\Sigma_A}$ that is an isometric map from the legs on $\Sigma_A$ to those on $A$. Here, the legs on $\Sigma_A$ means the legs emanating from $\Sigma_A$ into the blue region.
	
	\item Locally we consider any site on $\Sigma_A$ and analyze whether the tensor on this site represents an isometry from the legs outside $\mathcal{W}_{\Sigma_A}$ to those inside $\mathcal{W}_{\Sigma_A}$. In Fig.\ref{fig:3 cases}, the black dashed line is the PEE thread tangent to $\Sigma_A$. Excluding the two legs on the dashed black thread, we find half of the legs stays inside $\mathcal{W}_{\Sigma_A}$, while the other half stays outside. 
	
	\item If the two dashed legs also stay in $\mathcal{W}_{\Sigma_A}$, then the number of legs in  $\mathcal{W}_{\Sigma_A}$ is greater than the number of outside legs (see Fig.\ref{fig:case 1}), hence the perfect tensor on this site represents an isometric map from the outside legs to the inside legs. Since the composition of two isometric maps is also isometric, we extend $\mathcal{W}_{\Sigma_A}$ to incorporate this site. Consequently, the extended region still corresponds to an isometric map from the new boundary $\Sigma_A$ to $A$.
	
	\item On the other hand, if the two dashed legs stay outside $\mathcal{W}_{\Sigma_A}$ (see Fig.\ref{fig:case 3}), then the tensor is not an isometric map from the outside legs to the inside legs, hence we should not absorb the site into $\mathcal{W}_{\Sigma_A}$.
	
	\item Therefore, the critical configuration is that $\Sigma_A$ locally matches the dashed geodesic (see Fig.\ref{fig:case 2}).
	
	\item Let us start from the $\Sigma_A$ as the boundary interval $A$. Obviously it represents an isometry from the open legs to those stretch inside the AdS space. Then $\Sigma_A$ moves into the bulk by swallowing more and more bulk sites, until each site of $\Sigma_A$ matches the local critical configuration. This means the greedy algorithm halts when $\Sigma_A$ matches the RT surface $\gamma_{A}$, such that we arrive at the least upper bound. Furthermore, the entanglement wedge $\mathcal{W}_{\gamma_{A}} $ defines a isometric map from $\gamma_{A}$ to $A$ hence the upper bound \eqref{entropybound} is saturated.
\end{enumerate}

The advantage of our HaPPY-like PEE tensor network is that, the greedy algorithm is performed exactly in the AdS space instead of some discrete graph, and the minimal homologous surface exactly matches the RT surface. Nevertheless, as in the HaPPY code, the above arguments only works for connected regions and for a time slice of AdS$_3$.

\section{Random PEE tensor network}\label{sec:Random}

In the third model, we also start from the PEE network, but set the states $\ket{\T(\v{z})}$ \eqref{eq:bulk local tensor 2} on the bulk vertices to be unit random states following \cite{Hayden:2016cfa}. We will take the random average of such states when computing the entropy quantities. To be specific, we take an arbitrary reference state $\ket{0}_{\v{z}}$ and define $\ket{\T(\v{z})}=U_{\v{z}}\ket{0}_{\v{z}}$ where $U_{\v{z}}$ is a unitary operator, then the random average is equivalent to an integration over $U_{\v{z}}$ according to the Haar probability measure. Again, after contracting all the internal legs of adjacent bulk sites along all the PEE threads using \eqref{contraction0}, we obtain the random PEE tensor network representing a boundary state for the boundary open legs.

In \cite{Hayden:2016cfa}, the authors have shown that in a general graph of network with random tensors on the vertices, if we take the large $v$ limit for the bond dimension of legs, the entanglement entropy for an arbitrary boundary region $A$ is dominant by the area term:
\begin{equation}\label{eq:Sn}
	S_{A}= \min \mathcal{N}(\Sigma_A)\log v,
\end{equation}
where $\Sigma_A$ is a surface in the network homologous to $A$, and $\mathcal{N}(\Sigma_A)$ counts the number of contracted tensor legs cut by $\Sigma_A$. Note that, the above result applies to general network graphs, including our PEE network. In this case, counting the minimal number of contracted legs cut by $\Sigma_A$ is equivalent to computing the area of the minimal homologous surface $\gamma_{A}$, which is just the RT surface. 

Note that, in order to obtain the result \eqref{eq:Sn}, we need to set the contribution $\log v$ to the entropy from each thread to be a large number. Under this setup, we need to divide the density of the PEE threads by an overall constant $\log v$, such that the boundary PEE structure and the total crossing of the bulk surfaces are unchanged. In other words, the density of intersections between any bulk surface and the PEE network becomes $1/(4G \log v)$, such that $\min\mathcal{N}(\Sigma_A) =\text{Area}(\gamma_A)/(4G \log v)$. Hence we reproduced the RT formula,
\begin{equation}
	S_{A}= \min \mathcal{N}(\Sigma_A)\log v=\frac{\text{Area}(\gamma_A)}{4G}
\end{equation}
We emphasize that the result in \eqref{eq:Sn} applies to generic boundary regions, including disconnected and non-spherical connected regions. 

\section{Discussions}\label{sec:Discussions}

In this work, we assign quantum states to the vertices in the PEE network of AdS space and build three PEE tensor network models, where the entanglement entropy can be explicitly computed. In our models, counting the minimal number of cuts in the tensor network is exactly computing the area of the RT surface in the semi-classical geometry of gravity. This provides the first example of taking natural continuous limit for a discrete tensor network to perfectly match the geometric background. In our construction, the PEE network is determined by the PEE structure of the boundary CFT state; equivalently, it is built from the geodesics in the kinematic space, which itself is an intrinsic property of the background geometry. Thus, the network is not manually adjusted, and the RT formula is not used as an input.

In the factorized model, the RT formula is only understood for spherical regions. This is not surprising as in this model all the PEE threads are decoupled from each other, which means the model is oversimplified for AdS/CFT.  The factorized model reminds us of the bit thread configurations \cite{Freedman:2016zud,Headrick:2017ucz}, and we recommend the readers to consult \cite{Lin:2023rxc} for their similarities and differences\footnote{In \cite{Lin:2023rxc} it was demonstrated that bit threads configuration for spherical regions $A$ can be reproduced by the superposition of the vector fields describing the PEE threads emanating from $A$.}. 

More interesting models can be developed by including coupling among the PEE threads. In section \ref{sec:HaPPY}, besides regarding the tensors as tensor product of two-index tensors, we also construct a HaPPY-like \cite{Pastawski:2015qua} PEE tensor network, where we associate a perfect tensor to each site in the PEE network. The RT formula for connected regions in AdS$_3$ is derived using a similar greedy algorithm developed in \cite{Pastawski:2015qua} for our HaPPY-like PEE tensor network.

As in \cite{Hayden:2016cfa,Pastawski:2015qua}, we can also set open legs for the bulk sites to built a holographic code mapping the bulk states to the boundary states, then we can study the bulk quantum correction to the holographic entanglement entropy and identify the analogue of the quantum extremal surface \cite{Engelhardt:2014gca}. Perhaps the most interesting future direction is to generalize the PEE tensor network to geometric background beyond the Poincar\'e AdS space, thus construct a new framework to study holographic correspondence beyond AdS/CFT. This is nontrivial, as the Crofton formula requires essential modifications to apply to more general Riemannian manifolds. In such configurations, besides geodesics anchored on the boundary, one must also consider those that terminate at black hole singularities or end in the bulk.

	Our work may provide a new framework for studying quantum field theories on a smooth geometric background using continuous tensor network models, without losing the information of background geometry. This could be an advantage compared with the other continuous tensor network models (see also \cite{Quijandria:2014gsa,Tagliacozzo:2014bta,Draxler:2016nxp,Ganahl:2016mjv,DelasCuevas:2017yxc,Ganahl:2018jma,Cotler:2018ufx,Osborne:2019bsq,Tang:2020vyv,Tuybens:2020uth,Karanikolaou:2020rld,Stottmeister:2020ezd,Zou:2019xbi,Witteveen:2019lsk,Campos:2021zce,Osborne:2021ppp,Lukin:2022rds,Martyn:2022oll,Schmoll:2023eez,Vardian:2024jll,Chemissany:2025vye} for an incomplete list of recent developments). It is particularly interesting to apply this framework in flat space, where the Crofton formula holds and most non-gravitational field theories are studied.

\section*{}

\acknowledgments

The authors are supported by the NSFC Grant No. 12447108 and the Shing-Tung Yau Center of Southeast University. We thank Debarshi Basu for helpful discussions and related collaborations.

\appendix

\section{The RT formula in the factorized PEE tensor network}\label{sec:RTinfactoried}

We first introduce some preliminary about isometric tensor for the section to be self-contained.

\begin{defn}
	An isometry from $\H_A$ to $\H_B$ is a linear map $T:\H_A\rightarrow\H_B$ with the property that it preserves the inner product (i.e. $\braket{\psi}{\phi}=\braket{T\psi}{T\phi},~\forall \psi,\phi\in\H_A$), such that
	\begin{itemize}
		\item $\dim (A)\leq\dim (B)$;
		\item $T^\dagger T=\id_A,~TT^\dagger=\Pi_B$ where $\Pi_B$ is a projector on $Im(V)\subset \H_B$.
	\end{itemize}
\end{defn}
Note that a unitary map is also an isometry, in the case that $\dim (A)=\dim (B)$. Given a complete orthonormal basis $\{\ket{a}\}$ for $\H_A$, and $\{\ket{b}\}$ for $\H_B$, we represent an isometry $T$ as a two-index tensor:
\begin{equation}
	T:|a\rangle \mapsto T_{b a} |b\rangle,
\end{equation}
Then we have
\begin{equation}
	T_{ba^{\prime}}^{*} T_{b a}=T_{a^{\prime} b}^{\dagger} T_{b a}=\delta_{a^{\prime} a}
\end{equation}
which we call an isometric tensor. 

\begin{prop}\label{prop:isometry}
	Given an isometry $T:\left|a_1 a_2\right\rangle \mapsto T_{b a_1 a_2}|b\rangle $ and if $\H_A=\H_{A_1}\otimes\H_{A_2}$, then there exists an isometry (up to a constant) $\widetilde{T}: \mathcal{H}_{A_2} \rightarrow \mathcal{H}_B \otimes \mathcal{H}_{A_1}$ acting as 
	\begin{equation}
		\widetilde{T}:\left|a_2\right\rangle \mapsto T_{b a_1 a_2}\left|b a_1\right\rangle
	\end{equation}
	which obeys $\widetilde{T}^{\dagger} \widetilde{T}=\operatorname{dim}\left(A_1\right) I_{A_1}$.
\end{prop}

\begin{proof}
	\begin{equation}
		\begin{aligned}
			&T^*_{b a'_1 a'_2 }T_{b a_1 a_2}=\delta_{a'_1 a'_2,a_1 a_2}=\delta_{a'_1a_1}\delta_{a'_2a_2}\\
			\Rightarrow\quad &T_{b a_1 a'_2 }^* T_{b a_1 a_2} =\delta_{a_1a_1}\delta_{a_2'a_2}=\dim(A_1)\delta_{a'_1a_1}.
		\end{aligned}
	\end{equation}
\end{proof}

\begin{prop}\label{prop:orthonormal}
	Given an isometry $T:|a\rangle \mapsto T_{b a} |b\rangle$ and, $\{\ket{a}\}$ and $\{\ket{b}\}$ form complete orthonormal bases for $\H_A$ and $\H_B$ respectively, then
	\begin{equation}\label{eq:orthonomral}
		\{\ket{T_{a}}\equiv T\ket{a}=T_{b a} |b\rangle:1\leq a\leq \dim(A)\}
	\end{equation}
	forms a set of orthonormal vectors for $\H_B$. Note that the index $a$ ranges over $\dim(A) (\leq \dim(B))$ values, so that \eqref{eq:orthonomral} may not forms a basis for $\H_B$.
\end{prop}

\begin{proof}
	\begin{equation}
		\braket{T_{a'}}{T_a}=\diracprod{a'}{T^\dagger T}{a}=\diracprod{a'}{\id_A}{a}=\delta_{a^{\prime} a}.
	\end{equation}
	
\end{proof}

Then let us look at the building blocks of our factorized PEE tensor network, which are the two-index tensors $T_{ab}$ representing pairs of maximally entangled qudits. It implies that each tensor $T_{ab}$ correspond to a unitary map from one leg to the other. Contracting these two-index tensors along PEE threads results in geodesic chords with two-index tensors representing unitary maps from the qudit in one of the endpoints to another. This includes the two-index tensor $T(\v{x},\v{y})$ in \eqref{eq:unitary tensor} for a complete PEE thread. More explicitly, \eqref{eq:unitary tensor} implies
\begin{equation}
	T(\v{x},\v{y})T^\dagger(\v{x},\v{y})=\frac{\id_{\v{x}}}{v},\quad T^\dagger(\v{x},\v{y}) T(\v{x},\v{y})=\frac{\id_{\v{y}}}{v},
\end{equation} 
hence $T(\v{x},\v{y}):\ket{b}_{\v{y}}\mapsto T_{ab}(\v{x},\v{y})\ket{a}_{\v{x}}$ is a unitary map between the pair of qudits shown in Fig.\ref{fig:tensor at a PEE thread} at the two boundary sites $\v{x}$ and $\v{y}$. We can regard one leg of $T_{ab}(\v{x},\v{y})$ as input while the other as output, hence define a unitary flow direction for all the threads, i.e. each PEE thread $\mathcal{C}(\v{x},\v{y})$ has one incoming leg and one outgoing leg, giving a unitary map individually.

Now we consider any spherical boundary region $A$ and, as in the main text, classify the PEE threads into three types: the $\mathcal{C}^{A}$, $\mathcal{C}^{\bar{A}}$ and $\mathcal{C}^{A\bar{A}}$, which are represented by the gray, orange and blue semi-circles in Fig.\ref{fig:various PEE threads} respectively for the special case of AdS$_3$. Note that all the PEE threads are semi-circles and the RT surface $\gamma_{A}$ of $A$ is a semi-sphere, it is easy to observe that, all the PEE threads in $\mathcal{C}^{A}$ and $\mathcal{C}^{\bar{A}}$ do not intersect with $\gamma_{A}$, while any thread in  $\mathcal{C}^{A\bar{A}}$ intersects $\gamma_A$ exactly once. Next we cut the AdS space open along the RT surface $\gamma_{A}$, giving open legs living on $\gamma_{A}$. All the legs on $\gamma_{A}$ are all connected to a subset of the legs on $A$, which we call $A_1$. It is convenient to divide the rest of the legs in $A$ into two halves $A_2\cup A_3$ such that the gray threads connect the legs in $A_2$ to those in $A_3$. We further set the unitary flow direction for the threads $\mathcal{C}^{A\bar{A}}$ to be from the legs in $\gamma_{A}$ to those in $A_1$, and set the flow for the $\mathcal{C}^{A}$ threads to be from $A_2$ to $A_3$, then the threads in the entanglement wedge of $A$ defines a unitary evolution from the legs in $\gamma_{A}\cup A_3$ to the legs $A_1\cup A_2$.  Similarly, we can divide the legs in $\bar{A}$ into $\bar{A}=\bar{A}_1\cup \bar{A}_2 \cup\bar{A}_3$, and set $\bar{A}_1$  to be the legs connected to $\gamma_A$, then the PEE threads defines a unitary evolution from the legs in $\bar{A}_1\cup \bar{A}_2$ to those in $\gamma_{A}\cup \bar{A}_3$. The whole PEE tensor network is equivalent to a quantum circuit shown in Fig.\ref{fig:circuit}.

It is easy to see that, the open legs $A_2\cup A_3$ do not contribute to the entanglement entropy $S_A$, as they just correspond to a set of maximally entangled qudits which make a pure state. The only non-zero contribution to $S_A$ comes from the legs $A_1$, which are qudits maximally entangled to the legs $\bar{A}_1$. If we set the entanglement contained in each thread to be one, then $S_A$ is just given by the number of threads in $\mathcal{C}^{A\bar{A}}$, which equals to the number of intersections between $\gamma_A$ and the PEE threads. Since the density of intersections on $\gamma_{A}$ is $1/(4G)$ everywhere on $\gamma_{A}$, we reproduce the RT formula $S_A=\frac{\text{Area}(\gamma_{A})}{4G}$.

\begin{figure}[htbp]
	\centering
	\subfloat[]{\begin{tikzpicture}[scale=1.1]
			\draw[thick]
			(-2.7,0)--(4,0);
			\draw[line width=1.2pt,red,dashed]
			(2,0) arc (0:180:2);
			\tikzset{myarr/.style={-{Stealth[length=3pt,width=3pt,inset=0.2pt]}}}
			\draw[thick,orange!60]
			(3.8,0) arc (0:180:0.7);
			\draw[thick,orange!60]
			(2.5,0) arc (0:180:2.5);
			\draw[thick,myarr,orange!60] (3.8,0) arc (0:130:0.7);
			\draw[thick,myarr,orange!60] (2.5,0) arc (0:110:2.5);
			\draw[thick,black!60]
			(0.5,0) arc (0:180:0.5);
			\draw[thick,black!60]
			(1.8,0) arc (0:180:0.5);
			\draw[thick,black!60]
			(1.6,0) arc (0:180:1.6);
			\draw[thick,black!60]
			(-0.8,0) arc (0:180:0.5);
			\draw[thick,myarr,black!60] (0.5,0) arc (0:130:0.5);
			\draw[thick,myarr,black!60] (1.8,0) arc (0:130:0.5);
			\draw[thick,myarr,black!60] (-0.8,0) arc (0:110:0.5);
			\draw[thick,myarr,black!60] (1.6,0) arc (0:115:1.6);
			\draw[thick,blue!60]
			(-1.9,0) arc (0:80:0.4);
			\draw[thick,blue!60]
			(-1.5,0) arc (0:82:0.8);
			\draw[thick,blue!60]
			(-1.1,0) arc (0:84:1.2);
			\draw[thick,blue!60]
			(-0.7,0) arc (0:86:1.6);
			\draw[thick,blue!60]
			(-0.3,0) arc (0:88:2);
			\draw[thick,myarr,blue!60,-{Stealth[length=3pt,width=3pt,inset=0.2pt,reversed]}]
			(-1.9,0) arc (0:60:0.4);
			\draw[thick,myarr,blue!60,-{Stealth[length=3pt,width=3pt,inset=0.2pt,reversed]}] (-1.5,0) arc (0:70:0.8);
			\draw[thick,myarr,blue!60,-{Stealth[length=3pt,width=3pt,inset=0.2pt,reversed]}] (-1.1,0) arc (0:68:1.2);
			\draw[thick,myarr,blue!60,-{Stealth[length=3pt,width=3pt,inset=0.2pt,reversed]}] (-0.7,0) arc (0:65:1.6);
			\draw[thick,myarr,blue!60,-{Stealth[length=3pt,width=3pt,inset=0.2pt,reversed]}] (-0.3,0) arc (0:60:2);
			\draw[thick,blue!60]
			(1.9,0) arc (180:100:0.4);
			\draw[thick,blue!60]
			(1.5,0) arc (180:98:0.8);
			\draw[thick,blue!60]
			(1.1,0) arc (180:96:1.2);
			\draw[thick,blue!60]
			(0.7,0) arc (180:94:1.6);
			\draw[thick,blue!60]
			(0.3,0) arc (180:92:2);
			\draw[thick,myarr,blue!60,-{Stealth[length=3pt,width=3pt,inset=0.2pt,reversed]}] (1.9,0) arc (180:120:0.4);
			\draw[thick,myarr,blue!60,-{Stealth[length=3pt,width=3pt,inset=0.2pt,reversed]}] (1.5,0) arc (180:110:0.8);
			\draw[thick,myarr,blue!60,-{Stealth[length=3pt,width=3pt,inset=0.2pt,reversed]}] (1.1,0) arc (180:112:1.2);
			\draw[thick,myarr,blue!60,-{Stealth[length=3pt,width=3pt,inset=0.2pt,reversed]}] (0.7,0) arc (180:115:1.6);
			\draw[thick,myarr,blue!60,-{Stealth[length=3pt,width=3pt,inset=0.2pt,reversed]}] (0.3,0) arc (180:120:2);
			\filldraw[black] (2,0) circle (0.4pt);
			\filldraw[black] (2.4,0) circle (0.4pt);
			\filldraw[black] (2.5,0) circle (0.4pt);
			\filldraw[black] (-2.5,0) circle (0.4pt);
			\filldraw[black] (3.8,0) circle (0.4pt);
			\filldraw[black] (1.9,0) circle (0.4pt);
			\filldraw[black] (1.8,0) circle (0.4pt);
			\filldraw[black] (1.6,0) circle (0.4pt);
			\filldraw[black] (1.5,0) circle (0.4pt);
			\filldraw[black] (1.1,0) circle (0.4pt);
			\filldraw[black] (0.8,0) circle (0.4pt);
			\filldraw[black] (0.7,0) circle (0.4pt);
			\filldraw[black] (0.5,0) circle (0.4pt);
			\filldraw[black] (0.3,0) circle (0.4pt);
			\filldraw[black] (-2,0) circle (0.4pt);
			\filldraw[black] (-1.9,0) circle (0.4pt);
			\filldraw[black] (-1.8,0) circle (0.4pt);
			\filldraw[black] (-1.6,0) circle (0.4pt);
			\filldraw[black] (-1.5,0) circle (0.4pt);
			\filldraw[black] (-1.1,0) circle (0.4pt);
			\filldraw[black] (-0.8,0) circle (0.4pt);
			\filldraw[black] (-0.7,0) circle (0.4pt);
			\filldraw[black] (-0.5,0) circle (0.4pt);
			\filldraw[black] (-0.3,0) circle (0.4pt);
			\draw[] (0,-0.2) node {$A$};
			\draw[] (0,2.2) node {\textcolor{red}{$\gamma_A$}};
		\end{tikzpicture}
		\label{fig:various PEE threads}
	}
	\qquad\qquad
	\subfloat[]{
		\begin{tikzpicture}[scale=0.4]
			\draw[thick]
			(-5.5,0)--(6,0);
			\draw[thick]
			(-5.5,9)--(6,9);
			\draw[thick,blue!50]
			(0,0)--(0,9);
			\draw [-{Stealth[length=3pt,width=3pt,inset=0.2pt,reversed]},blue!50,thick] (0,0) -- (0,3);
			\draw [-{Stealth[length=3pt,width=3pt,inset=0.2pt,reversed]},blue!50,thick] (1,0) -- (1,3);
			\draw [-{Stealth[length=3pt,width=3pt,inset=0.2pt,reversed]},blue!50,thick] (2,0) -- (2,3);
			\draw [-{Stealth[length=3pt,width=3pt,inset=0.2pt,reversed]},blue!50,thick] (0,0) -- (0,6);
			\draw [-{Stealth[length=3pt,width=3pt,inset=0.2pt,reversed]},blue!50,thick] (1,0) -- (1,6);
			\draw [-{Stealth[length=3pt,width=3pt,inset=0.2pt,reversed]},blue!50,thick] (2,0) -- (2,6);
			\draw[thick,blue!50]
			(1,0)--(1,9);
			\draw[thick,blue!50]
			(2,0)--(2,9);
			\draw[thick,black!60]
			(3.5,0)--(3.5,9);
			\draw[thick,black!60]
			(4.5,0)--(4.5,9);
			\draw[thick,black!60]
			(5.5,0)--(5.5,9);
			\draw [-{Stealth[length=3pt,width=3pt,inset=0.2pt,reversed]},black!60,thick] (3.5,0) -- (3.5,3);
			\draw [-{Stealth[length=3pt,width=3pt,inset=0.2pt,reversed]},black!60,thick] (4.5,0) -- (4.5,3);
			\draw [-{Stealth[length=3pt,width=3pt,inset=0.2pt,reversed]},black!60,thick] (5.5,0) -- (5.5,3);
			\draw [-{Stealth[length=3pt,width=3pt,inset=0.2pt,reversed]},black!60,thick] (3.5,0) -- (3.5,6);
			\draw [-{Stealth[length=3pt,width=3pt,inset=0.2pt,reversed]},black!60,thick] (4.5,0) -- (4.5,6);
			\draw [-{Stealth[length=3pt,width=3pt,inset=0.2pt,reversed]},black!60,thick] (5.5,0) -- (5.5,6);
			\draw[thick,orange!60]
			(-2.5,0)--(-2.5,9);
			\draw[thick,orange!60]
			(-3.5,0)--(-3.5,9);
			\draw[thick,orange!60]
			(-4.5,0)--(-4.5,9);
			\draw [-{Stealth[length=3pt,width=3pt,inset=0.2pt,reversed]},orange!60,thick] (-2.5,0) -- (-2.5,3);
			\draw [-{Stealth[length=3pt,width=3pt,inset=0.2pt,reversed]},orange!60,thick] (-3.5,0) -- (-3.5,3);
			\draw [-{Stealth[length=3pt,width=3pt,inset=0.2pt,reversed]},orange!60,thick] (-4.5,0) -- (-4.5,3);
			\draw [-{Stealth[length=3pt,width=3pt,inset=0.2pt,reversed]},orange!60,thick] (-2.5,0) -- (-2.5,6);
			\draw [-{Stealth[length=3pt,width=3pt,inset=0.2pt,reversed]},orange!60,thick] (-3.5,0) -- (-3.5,6);
			\draw [-{Stealth[length=3pt,width=3pt,inset=0.2pt,reversed]},orange!60,thick] (-4.5,0) -- (-4.5,6);
			\draw[thick,blue!50]
			(0,4 )--(0,5);
			\draw[thick,blue!50]
			(2,4)--(2,5);
			\draw[thick,black!60]
			(3.5,4)--(3.5,9);
			\draw[thick,black!60]
			(5.5,4)--(5.5,9);
			\draw[thick,orange!60]
			(-4.5,7.5)--(-4.5,9);
			\draw[thick,orange!60]
			(-2.5,7.5)--(-2.5,9);
			\draw[thick,blue!50]
			(0,7.5)--(0,9);
			\draw[thick,blue!50]
			(2,7.5)--(2,9);
			\draw[red, dashed,line width=1.2pt, rounded corners=5mm]
			(-1,0)  -- (-1,4.5) -- (3,4.5) -- (3,9);
			\draw[] (1.0,-0.6) node {$A_1$};
			\draw[] (4.5,-0.6) node {$A_2$};
			\draw[] (-3.5,-0.6) node {$\bar{A}_3$};
			\draw[] (-3.5,9.55) node {$\bar{A}_2$};
			\draw[] (4.5,9.55) node {$A_3$};
			\draw[] (1.0,9.55) node {$\bar{A} _1$};
			\draw[] (-1.6,2.5) node {\textcolor{red}{$\gamma_A$}};
		\end{tikzpicture}	
		\label{fig:circuit}
	}
	\caption{(a) The three type of PEE threads, denoted by  $\mathcal{C}^{A}$, $\mathcal{C}^{\bar{A}}$ and $\mathcal{C}^{A\bar{A}}$, are labeled by gray, orange and blue solid simi-circles respectively, while the RT surface $\gamma_{A}$ is drawn as the red dashed line; (b) The circuit interpretation of the factorized PEE tensor network. Each PEE thread individually gives a unitary map, and the RT surface $\gamma_{A}$ only cuts the threads in   $\mathcal{C}^{A\bar{A}}$ once. The graph can be compared with Fig.19 of \cite{Pastawski:2015qua} where the ``circuit'' is built from the net effect of perfect tensors of certain bulk subregions as unitary gates.}
\end{figure}

It is also useful to derive the RT formula using the arguments in \cite{Pastawski:2015qua} for more generic tensor network models. Again, let us cut the tensor network open along a homologous surface $\mathcal{C}$ whose boundary matches the boundary of $A$, such that the tensor network is divided into two parts $\P$ and $\Q$, which are maps from $\mathcal{C}$ to $A$ and from $\mathcal{C}$ to $\bar{A}$ respectively. According to the proposition \eqref{prop:orthonormal}, if the tensors $\P$ and $\Q$ are isometries, then $\ket{\Q_{\v{c}}}$ and $\ket{\P_\v{c}}$ are two sets of orthonormal vectors. This implies $\rho_{A}$ is proportional to the identity such that the upper bound \eqref{entropybound} for entanglement entropy is saturated. 

In the case of our factorized PEE tensor network model, we have:
\begin{itemize}
	\item the homologous surface with minimal number of cuts is exactly the RT surface $\gamma_{A}$, and $\mathcal{N}(\gamma_{A})={\text{Area}(\gamma_{A})}/{4G}$,
	\item by regarding $\{\gamma_{A},A_3,A_1\cup A_2\}$ as $\{a_2,a_1,b\}$ in the proposition \ref{prop:isometry}, we can construct an isometry $\P$ that maps $\gamma_A$ to $A$, likewise we can construct an isometry $\Q$ from $\gamma_{A}$ to $\bar{A}$.
\end{itemize}
Therefore, the inequality is saturated and the RT formula follows upon setting $\log v=1$.

\section*{}

\bibliographystyle{JHEP}
\bibliography{bib}
\end{document}